\documentclass[journal]{IEEEtran}
\ifCLASSINFOpdf
\else
\fi
\usepackage{graphicx}
\usepackage{cite}
\usepackage{colortbl}
\usepackage{textcase}
\usepackage{amsmath}
\usepackage{amsfonts}

\newtheorem{theorem}{Theorem}
\newtheorem{remark}{Remark}
\newtheorem{lemma}{Lemma}
\newtheorem{proof}{Proof}

\begin{document}
%
\title{Sampled-data in Space Control of Scalar Semilinear Parabolic and Hyperbolic Systems}
%
%
%

\author{Igor~Furtat,
        Pavel~Gushchin
\thanks{I. Furtat was with the Institute for Problems in Mechanical Engineering
of the Russian Academy of Sciences, V.O., Bolshoj pr., 61, St. Petersburg, 199178
Russia e-mail: cainenash@mail.ru.}
\thanks{P. Gushchin is with the Gubkin University, 65 Leninsky Prospekt, Moscow, 119991, Russia.}}
\maketitle

\begin{abstract}
The paper describes a novel method of sampled-data in space (spatial variable) control of  scalar semilinear systems of parabolic and hyperbolic type with unknown parameters and distributed disturbances.
A finite set of sampled-data in the spatial variable measurements is available.
The control law depends on the function which depends on the spatial variable and on a finite set of state measurements.
A special choice of this function can affect on some properties of the closed-loop system.
In particular, the paper describes the examples of this function that provides reduced control costs in comparison with some other control methods.
The exponential stability of the closed-loop system and robustness with respect to unknown parameters and disturbances is proposed in terms of linear matrix inequalities (LMIs).
The simulations confirm theoretical results and show the efficiency of the proposed control law compared with some existing ones.
\end{abstract}


%
\IEEEpeerreviewmaketitle

\section{Introduction}

The paper considers systems presented by semilinear differential equations of parabolic and hyperbolic type with distributed disturbances.
Such systems describe convection-diffusion processes, a rotating column of a compressor with an air injection drive, heat distribution in a rod, string oscillation, etc.

The finite-dimensional control using Fourier transform and Galerkin method are considered in \cite{Hagen03,Smagina06,Candogan08,Zhong18}.
In \cite{Demetriou10} a control law based on moving sensors and actuators along the spatial variable is proposed for linear parabolic systems.
Also for such systems the adaptive control based on the backstepping method is proposed in \cite{Smyshlyaev05,Krstic08,Izadi15,Chen19}. 
However, the backstepping control law is complicated in calculation and implementation.

Differently from \cite{Hagen03,Smagina06,Candogan08, Demetriou10,Smyshlyaev05,Krstic08,Izadi15,Zhong18,Chen19}, in the present paper we propose a method for design a sampled-data in space control law.
For finite-dimensional systems a similar approach has been studied over the past few decades as a discretization (quantization) of measured signal, see, e.g. \cite{Delchamps89,Brockett00,Baillieul02,Zheng12,Furtat15,Li20}.
Unlike continuous control law such discrete one does not take into account the behavior of the plant between samples,
but in some cases it allows solving a number of technical problems, e.g. control via digital communication channels, control with restriction on information communication channels, etc.
In our paper the spatial sampling is used to obtain the implementable control law.

The observability of distributed systems with sampled-data space is studied in \cite{Khapalov93}.
The sampled-data in space control of infinite-dimensional systems with known parameters is considered in \cite{Cheng09, Logemann05}.
Differently from \cite{Cheng09,Logemann05} in \cite{Fridman12,Fridman14,Fridman20} a sampled-data in space control of parabolic systems with unknown distributed parameters is considered. 
Also in \cite{Fridman12,Fridman14,Fridman20} the analysis of exponential stability of the closed-loop system is proposed in terms linear matrix inequalities (LMIs).
However, the results of \cite{Cheng09,Logemann05,Fridman12,Fridman14,Fridman20} are obtained without disturbances.

In the present paper, as in \cite{Cheng09,Logemann05,Fridman12,Fridman14,Fridman20}, we propose the sampled-data in space control law.
Differently from \cite{Cheng09,Logemann05,Fridman12,Fridman14,Fridman20}, the main contribution of our paper is as follows:
\begin{enumerate}
\item[(i)] we consider scalar semiliniar parabolic and hyperbolic systems under distributed disturbances and unknown distributed parameters;
\item[(ii)] the proposed control law allows one to form various configurations of the control signal in a spatial variable in order to obtain various properties in the closed-loop system. In particular, stabilization with less control costs is considered;
\item[(iii)] the exponential stability of the closed-loop systems under disturbances and unknown parameters using the LMI approach is proposed.
\end{enumerate}

The paper is organized as follows. 
Problem formulation is presented in Section \ref{Sec2}.
Section \ref{Sec3} describes the control law design.
Sections \ref{Sec4}, \ref{Sec04} consider the analysis of exponential stability of the closed-loop systems in terms of LMIs.
The well-posedness problem is considered in Section \ref{Sec5}. 
Section \ref{Examples} illustrates an efficiency of the proposed method and its advantages compared with the existing methods. 
Section \ref{Sec7} collects some conclusions.

\textit{Notation}. Throughout the paper
$\mathbb R^{n}$ denotes the $n$ dimensional Euclidean space with the norm $|\cdot|$; 
$\mathbb R^{n \times m}$ is the set of all $n \times m$ real matrices; 
$P>0$ and $P \in \mathbb R^{n \times n}$ means that $P$ is symmetric and positive definite; 
the symmetric elements of the symmetric matrix will be denoted by $*$. 
Functions, continuously differentiable in all arguments, are referred to as of class $\mathcal{C}^1$.
Subscripts indicate partial derivatives $z_\xi = \frac{\partial z}{\partial \xi}$ and $z_{\xi \xi}=\frac{\partial^2 z}{\partial \xi^2}$.
$L_2(0, l)$ is the Hilbert space of square integrable functions $z(\xi )$, $\xi \in [0, l]$ with the corresponding  norm $ \| z \|_{L_2}^2=\int_0^l z^2(s) ds$.
$H_1(0, l)$ is the Sobolev space of absolutely continuous scalar functions $z: [0,l]  \to \mathbb R$ with the norm $\|z\|_{H_1}^2=\int_{0}^l z_s^2(s)ds$ and $z_\xi \in L_2(0, l)$.
$H_2(0, l)$ is the Sobolev space of scalar functions $z: [0,l] \in \mathbb R$ with absolutely continuous $z_{\xi}$, the norm $\|z\|_{H_2}^2=\int_{0}^l z_{ss}^2(s)ds $ and $z_{\xi\xi} \in L_2(0,l)$.

\section{Problem formulation} \label{Sec2}

\subsection{Models}

1) Consider a semilinear scalar equation of parabolic type in the form
\begin{equation}
\label{eq1}
\begin{array}{l}
z_{t}(x,t)=\frac{\partial}{\partial x} [a_1(x) z_{x}(x,t)] + a_2(x)z_x(x,t)
\\
~~~~~~~~~~~~~
+ \phi(z,x,t) z(x,t)+u(x,t)+f(x,t),
\\
~~~~~~~~~~~~~
x \in [0,l],~l>0,
\end{array}
\end{equation} 
with Dirichlet boundary conditions
\begin{equation}
\label{eq1_b1}
\begin{array}{l}
z(0,t)=z(l,t)=0
\end{array}
\end{equation} 
or with mixed boundary conditions
\begin{equation}
\label{eq1_b2}
\begin{array}{l}

z_x(0,t)=\gamma z(0,t),~~z(l,t)=0,~~\gamma \geq 0.
\end{array}
\end{equation}
Here 
$t \geq 0$,
$z: [0,l] \times [0,\infty) \to \mathbb R$ is the state,
$u(x,t)$ is the control. 
The functions $a_1(x)$, $a_2(x)$, $\phi(z(x,t),x,t)$ and $f(x,t)$ are unknown and of class $\mathcal{C}^1$.
Also these functions satisfy the following conditions: $a_1(x) \geq \underline{a}_1 > 0$,  $\underline{a}_2 \leq a_2(x) \leq \overline{a}_2$, 
$\underline{\phi} \leq \phi(z,x,t) \leq \overline{\phi}$, 
$|f(x,t)| \leq \bar{f}$ with known bounds $\underline{a}_1$, $\underline{a}_2$, $\overline{a}_2$, $\underline{\phi}$, $\overline{\phi}$, and $\bar{f}$.
The value of $\gamma$ in \eqref{eq1_b2} may be unknown.

\begin{remark}
\label{Rem01}
System \eqref{eq1} describes convection-diffusion processes under $u(x,t)=0$. 
In \cite{Hagen03} system  \eqref{eq1} describes a rotating stand of a compressor with an air injection drive $u(x,t)$, where $z(x,t)$ is the axial flow through the compressor. 
Also the boundary-value problem \eqref{eq1}, \eqref{eq1_b1} with $u(x,t)=0$ describes the heat distribution in a uniform one-dimensional rod with a fixed temperature at the ends, where $a_1$ and $\phi$ are coefficients of thermal conductivity and heat transfer with the environment respectively, $a_2=0$, $z(x,t)$ is the temperature at time $t$ in the point $x$.
\end{remark}


2) Additionally, consider a semilinear scalar differential equation of hyperbolic type in the form
\begin{equation}
\label{eq01}
\begin{array}{l}
z_{tt}(x,t)=\frac{\partial}{\partial x} [a_1(x) z_{x}(x,t)] + a_2(x) z_{x}(x,t)
\\
~~~~~~~~~~~~~
 + \phi(z,x,t) z(x,t)
\\
~~~~~~~~~~~~~
 - b(z,x,t) z_t(x,t)
+u(x,t)+f(x,t),
\\
~~~~~~~~~~~~~
x \in [0,l],
\end{array}
\end{equation} 
with Dirichlet boundary conditions \eqref{eq1_b1}
or with mixed boundary conditions
\eqref{eq1_b2}.
In  \eqref{eq01} the function $b(z,x,t)$ is of class $\mathcal{C}^1$ which satisfies the condition
 $0<\underline{b} \leq b(z,x,t) \leq \overline{b}$ with known bounds $\underline{b}$ and $\overline{b}$. 
 Other functions in \eqref{eq01} take the same values as in \eqref{eq1}.

\begin{remark}
\label{Rem02}
The boundary-value problem  \eqref{eq01}, \eqref{eq1_b1} with $u(x,t)=0$ describes vibrations of a uniform string with fixed ends and energy dissipation, where $a_1$, $b$ and $\phi$ are coefficients of elasticity, dissipation and stiffness respectively, $a_2=0$, $z(x, t)$ and $z_t(x, t)$ are deflection and speed of the string at time $t$ in the point $x$.
\end{remark}


\subsection{Objective}

Divide the segment $[0,l]$ into $N$  sampling intervals (not necessarily equal length) and denote
\begin{equation}
\label{eq1_sub_int}
\begin{array}{l}
0=x_0<x_1<...<x_N=l, ~~~ \Delta \geq x_{j+1}-x_j,
\\
j=0,...,N-1.
\end{array}
\end{equation}
Here $\Delta$ is known value.
Assume that $N$ sensors are placed inside these sub-intervals, i.e. only the signals $z(\bar{x}_j,t)$, $\bar{x}_j \in (x_j, x_{j-1})$,  $j=0,...,N-1$ are available for measurement.

Our objective is to design the control law (sampled-data in space) which ensures the exponential stability of the closed-loop system for \eqref{eq1} or \eqref{eq01}.


\section{Control law design} \label{Sec3}

Introduce the control law in the form
\begin{equation}
\label{eq_g_c_l}
\begin{array}{l}
u(x,t)= - K F^j(x,t),
~~
x \in [x_{j}, x_{j+1}), 
\\
~~~~~~~~~~~~
j=0,...,N-1,
\end{array}
\end{equation}
where  $K > 0$, 
the function $F^j(x,t)$ satisfies the following conditions:
\begin{enumerate}
\item[(i)] $F^j(x,t)$ is of class $\mathcal{C}^1$;
\item[(ii)] $F^j(\bar{x}_j,t)=z(\bar{x}_j,t)$, where $\bar{x}_j \in (x_j, x_{j-1})$;
\item[(iii)] $F_x^j(x,t)$ is bounded for any $x \in [0,l]$ and $t \geq 0$.
\end{enumerate}

The condition (i) will be required to solve the boundary-value problem in Section \ref{Sec5}. The conditions (ii) and (iii) will be required to proof the closed-loop system stability in Sections \ref{Sec4} and \ref{Sec04}.

Consider a few examples of the function $F^j(x,t)$.

\textit{Example 1.}
Let $F^j(x,t)=\varphi^j(x,t) z(\bar{x}_j,t)$, where $j=0,...,N-1$, $\varphi^j(x,t)$ is of class $\mathcal{C}^1$, $\varphi_x^j(x,t) z(\bar{x}_j,t)$ is bounded for any $x \in [x_{j}, x_{j+1})$, $t \geq 0$ and $\varphi^j(\bar{x}_j,t)=1$.


\textit{Example 2.}
In \cite{Fridman12,Fridman14,Fridman20} $u(x,t) \neq 0$ over the interval $[x_{j}, x_{j+1})$ when $z(\bar{x}_j,t) \neq 0$. 
Now we consider a case when $u(x,t) \neq 0$ only on a part of the interval $[x_{j}, x_{j+1})$ for $z(\bar{x}_j,t) \neq 0$.
Let  in example 1 the function $\varphi^j(x,t)$ is given in the form
\footnotesize
\begin{equation}
\label{fff}
\begin{array}{l}
\varphi^j(x,t)=
\begin{cases}
   0.5+0.5\cos\left(\frac{\alpha (x-\bar{x}_j)}{1+z^2(\bar{x}_i,t)}\right) ~~\mbox{if} 
   \\ 
   x \in \left[\bar{x}_j-\frac{\pi (1+z^2(\bar{x}_j,t))}{\alpha};\bar{x}_j+\frac{\pi (1+z^2(\bar{x}_j,t))}{\alpha}\right]; \\
   0 ~~\mbox{if} 
   \\ 
   x \in \left[x_j; \bar{x}_j-\frac{\pi (1+z^2(\bar{x}_j,t))}{\alpha}\right) \\ \cup \left(\bar{x}_j+\frac{\pi (1+z^2(\bar{x}_j,t))}{\alpha}; x_{j+1} \right).
 \end{cases}
 \end{array}
\end{equation}
\normalsize
Here $\alpha > 0$ is a sufficiently large number that can be selected from the condition $\alpha > \max \left \{\frac{\bar{x}_j-x_{j}}{\pi(1+z^2(\bar{x}_j,t))},\frac{x_{j+1}-\bar{x}_j}{\pi(1+z^2(\bar{x}_j,t))} \right\}$. 
It is obvious, that $\varphi^j(\bar{x}_j,t)=1$,  $\varphi^j(x,t)$ is of class $\mathcal{C}^1$, $\varphi_x^j(x,t) z(\bar{x}_j,t)$ is bounded for any $x \in [x_{j}, x_{j+1})$, $t \geq 0$ and $z(\bar{x}_j,t) \in \mathbb R$. 

\begin{figure}[h]
\center{\includegraphics[width=0.7\linewidth]{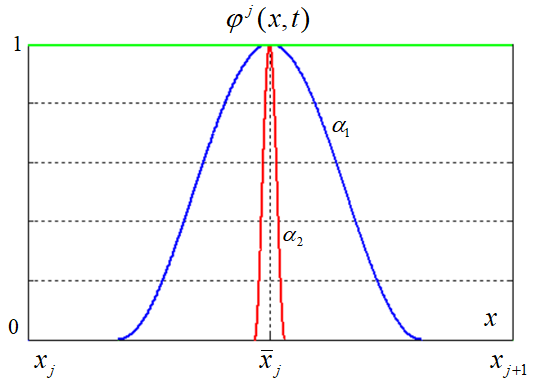}}
\caption{The plots of $\varphi^j(x,t)$ from example 2 (for $z(\bar{x}_j,t)=const > 0$) on the interval $x \in [x_{j}, x_{j+1})$ at various values of $\alpha>0$ ($(\alpha=\alpha_1)<(\alpha=\alpha_2)$).}
\label{Fig02}
\end{figure}


\textit{Example 3.} 
In \eqref{fff} transition between the values of $\varphi^j$ depends on $z(\bar{x}_j,t)$.
Now let us consider new function that eliminates this dependency:
\small
\begin{equation}
\label{fff11}
\begin{array}{l}
\varphi^j(x,t)=
\begin{cases}
  e^{-\frac{\alpha}{(1+z^2(\bar{x}_j,t))(\beta_j^2-(x-\bar{x}_j)^2)}+\frac{\alpha}{\beta_j(1+z^2(\bar{x}_j,t))}}   \\ \mbox{if} ~~ x \in \left(\bar{x}_j-\beta_j;\bar{x}_j+\beta_j\right); \\
   0 ~~\mbox{if} \\ x \in \left[x_j; \bar{x}_j-\beta_j\right] \cup \left[\bar{x}_j+\beta_j; x_{j+1} \right).
 \end{cases}
 \end{array}
\end{equation}
\normalsize
Here $\alpha > 0$, $\beta_j \leq \min\{x_{j+1}-\bar{x}_j,\bar{x}_j-x_{j}\}$. 
Differently from \eqref{fff}, transition between the function values in \eqref{fff11} does not depend on $z(\bar{x}_j,t)$, 
but the transition depends only on $\beta_j$. 
We have $\varphi^j(z(\bar{x}_j,t),\bar{x}_j,t)=1$,  $\varphi^j(z(\bar{x}_j,t),x,t)$ is of class $\mathcal{C}^1$ and $\varphi_x^j(z(\bar{x}_j,t),x,t) z(\bar{x}_j,t)$ is bounded for any $x \in [x_{j}, x_{j+1})$, $t \geq 0$ and $z(\bar{x}_j,t) \in \mathbb R$. 
The graphs of \eqref{fff11} are given in Fig.~\ref{Fig012}.

\begin{figure}[h]
\center{\includegraphics[width=0.7\linewidth]{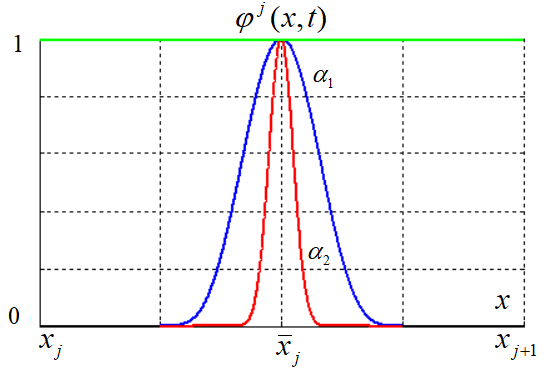}}
\caption{The plots of $\varphi^j(z(\bar{x}_j,t),x,t)$ from Example 3 (for $z(\bar{x}_j,t)=const > 0$)  on the interval $x \in [x_{j}, x_{j+1})$ at various values of $\alpha>0$ ($(\alpha=\alpha_1)<(\alpha=\alpha_2)$).}
\label{Fig012}
\end{figure}


\begin{remark}
\label{Rem1}
In \cite{Fridman12,Fridman14,Fridman20} the control law $u(x,t)= - K z(\bar{x}_j,t)$, $\bar{x}_j = 0.5 (x_j + x_{j+1})$, $x \in [x_{j}, x_{j+1})$, 
$j=0,...,N-1$ is proposed. 
This control law is a special case of \eqref{eq_g_c_l} when $F^j(x,t)=z(\bar{x}_j,t)$ (or $\varphi^j(x,t)=1$ in example 1). 
If the function $F^j(x,t)$ is chosen as in examples 2 and 3, then the area under the curve $\varphi^j(x,t)=1$  can be significantly larger than the area under the proposed curves (see Figs.~\ref{Fig02}, \ref{Fig012}). Thus, the proposed control law can stabilize system \eqref{eq1} at a lower cost of the control signal.
\end{remark}


\section{Main result for a parabolic type system \eqref{eq1}} \label{Sec4}

Substituting \eqref{eq_g_c_l} into \eqref{eq1}, represent the closed-loop system as follows
\begin{equation}
\label{eq_cl_loop}
\begin{array}{l}
z_{t}(x,t)=\frac{\partial}{\partial x} [a_1(x) z_{x}(x,t)]+a_2(x)z_x(x,t)
\\
~~~~~~~~~~~~~ +f(x,t)-(K-\phi(z,x,t)) z(x,t) 
\\
~~~~~~~~~~~~~
+ K[z(x,t) - F^j(x,t)],
\\
~~~~~~~~~~~~~ x \in [x_{j}, x_{j+1}), ~~~ j=0,...,N-1.
\end{array}
\end{equation}

\begin{theorem}
\label{Th1}
Consider the closed-loop system \eqref{eq_cl_loop} under the boundary conditions \eqref{eq1_b1} or \eqref{eq1_b2}.
Given $\bar{R}>0$, $\Delta>0$ and $\delta>0$ there exists $K>0$ such that the following two linear matrix inequalities hold
\begin{equation}
\label{eq_LMI0}
\begin{array}{l}
\Psi (a_2=\underline{a}_2) \leq 0 ~~~ \mbox{and} ~~~ \Psi (a_2=\overline{a}_2) \leq 0,
\end{array}
\end{equation}
where
\begin{equation}
\label{eq_LMI}
\begin{array}{l}
\Psi=
\begin{bmatrix}
\Psi_{11} & a_2 & 1 & 0 \\
* & \Psi_{22} & 0 & -\frac{4 \Delta^2}{\pi^2} K \bar{R}^{-1} \\
* & * & -\beta_1 & 0 \\
* & * & * & -\beta_2 + \frac{4 \Delta^2}{\pi^2} K \bar{R}^{-1}
\end{bmatrix},
\\
\Psi_{11}=-2K+2\overline{\phi} + 2\delta + K \bar{R}, 
\\
\Psi_{22}=-2\underline{a}_1 + \frac{4 \Delta^2}{\pi^2} K \bar{R}^{-1}. 
\end{array}
\end{equation}
Then the following inequality holds
\begin{equation}
\label{eq_th2}
\begin{array}{l}
\| z(\cdot,t) \|^2_{L_2}  
\leq 
e^{-2\delta t} \| z(\cdot,0) \|^2_{L_2} 
+\frac{\gamma}{2\delta},
\end{array}
\end{equation}
where
\[\gamma=\beta_1 \sup\limits_{t \geq 0} \int_0^{l}f^2(x,t)dx + \beta_2 \sup\limits_{t \geq 0}  \sum_{j=0}^{N-1} \int_{x_j}^{x_{j+1}} (F_x^j(x,t))^2 dx.\]

\end{theorem}


Before proving this theorem, consider two auxiliary lemmas.

\begin{lemma} (Extended Wirtinger inequality).
\label{lem_Har_ext}
Let $z \in H_1(0,l)$ be a scalar function, $0 = \chi_0 < \chi_1 < ... \chi_{n-1} < \chi_n = l$ and $\Delta \geq \chi_i - \chi_{i+1}$, $i=0,...,n-1$. If $z(\chi_i)=0$, $i=1,...,n-1$, then
\begin{equation}
\label{eq_Har10}
\begin{array}{l}
\int_0^l z^2(\xi) d\xi \leq \frac{4\Delta^2}{\pi^2} \int_0^l z_{\xi}^2 (\xi) d\xi.
\end{array}
\end{equation} 
\end{lemma}

\begin{proof}
Rewrite the left hand side of \eqref{eq_Har10} in the form
\begin{equation*}
\label{eq_Har100}
\begin{array}{l}
\int_0^l z^2(\xi) d\xi 
=\int_0^{\chi_1} z^2(\xi) d\xi 
\\
~~~~~~~~~~~~~~~
+\sum_{i=1}^{n-1}\int_{\chi_i}^{\chi_{i+1}} z^2(\xi) d\xi 
+ \int_{\chi_n}^l z^2(\xi) d\xi.
\end{array}
\end{equation*} 
Using Wirtinger's inequality, we obtain
\begin{equation*}
\label{eq_Har100}
\begin{array}{l}
\int_0^{\chi_1} z^2(\xi) d\xi +  
\int_{\chi_n}^l z^2(\xi) d\xi + 
\sum_{i=1}^{n-1}\int_{\chi_i}^{\chi_{i+1}} z^2(\xi) d\xi \leq
\\
\leq 
\frac{4 \Delta^2}{\pi^2} \left(\int_0^{\chi_1} z_{\xi}^2 (\xi) d\xi 
+ \int_{\chi_n}^l z_{\xi}^2 (\xi) d\xi \right)
\\
+\frac{\Delta^2}{\pi^2}  \sum_{i=1}^{n-1} \int_{\chi_i}^{\chi_{i-1}} z_{\xi}^2 (\xi) d\xi

\leq \frac{4\Delta^2}{\pi^2} \int_0^l z_{\xi}^2 (\xi) d\xi.
\end{array}
\end{equation*} 
\end{proof}

\begin{lemma}
\label{lem_Hal}
Let the function $V: [t_0, \infty) \to [0, \infty)$ be differentiable on $[t_0, \infty)$, $t_0 \geq 0$. 
Consider the following differential inequality
\begin{equation}
\label{eq_Hal1}
\begin{array}{l}
\dot{V}(t) \leq -\delta V(t) + f(t),
\end{array}
\end{equation} 
where $\delta>0$ and $\sup\limits_{t \geq t_0}|f(t)| = \beta$. 
Then the following inequality holds
\begin{equation}
\label{eq_Hal2}
\begin{array}{l}
V(t) \leq e^{-\delta (t-t_0)} V(t_0) + \frac{\beta}{\delta}, ~~~ t \geq t_0.
\end{array}
\end{equation} 
\end{lemma}

\begin{proof}
Denote by
\begin{equation*}
\label{eq_Hal3}
\begin{array}{l}
y(t)=e^{-\delta (t-t_0)} V(t_0) +\frac{\beta}{\delta}, ~~~ t \geq t_0.
\end{array}
\end{equation*}
It is easy to verify that the function $y(t)$ is a solution of the differential equation
\begin{equation}
\label{eq_Hal4}
\begin{array}{l}
\dot{y}(t) = -\delta y(t) + \beta, ~~~ t \geq t_0.
\end{array}
\end{equation} 
Using the comparison principle, we will show that $V(t) \leq y(t)$ for any $t \geq t_0$. Let $\varepsilon_1 > \varepsilon_2 > ... > \varepsilon_n > ...$ is a sequence of positive numbers such that $\lim\limits_{n \to \infty} \varepsilon_n = 0^{+}$. Then the function
\begin{equation}
\label{eq_Hal5}
\begin{array}{l}
\dot{y}_n(t) = -\delta y(t) + \frac{\beta}{\delta} + \frac{\varepsilon_n}{\delta}
\end{array}
\end{equation} 
is a solution of the differential equation
\begin{equation}
\label{eq_Hal6}
\begin{array}{l}
\dot{y}(t) = -\delta y(t) + \beta+ \varepsilon_n.
\end{array}
\end{equation} 
Suppose that there exists $t^*>t_0$ such that
\begin{equation}
\label{eq_Hal7}
\begin{array}{l}
t^*=\inf\{t>t_0: V(t) \geq y_n(t) \}.
\end{array}
\end{equation} 
Then $V(t^*) \geq y_n(t^*)$ and $V(t)<y_n(t)$ at $t_0 \leq t \leq t^*$.
The condition $\dot{V}(t^*) < \dot{y}_n(t^*)$ holds from \eqref{eq_Hal1} and \eqref{eq_Hal6}. 
On the other hand, the condition $\dot{V}(t^*) \geq \dot{y}_n(t^*)$ follows from $V(t) < y_n(t)$ at $t<t^*$ and $V(t^*)=y_n(t^*)$.
We have a contradiction. 
 Therefore, $V(t) < y_n(t)$ for any $t \geq t_0$ and $n=1,2,...$ 
 Consequently, $V(t) \leq \lim\limits_{n \to \infty} y_n(t)=y(t)$ for all $t \geq t_0$. 
 Lemma \ref{lem_Hal} is proved.
\end{proof}

\begin{proof} 
To analyze the stability of the closed-loop system \eqref{eq_cl_loop}, consider the following Lyapunov functional
\begin{equation}
\label{eq_L_F}
\begin{array}{l}
V(t)=\int_{0}^{l} z^2(x,t) dx.
\end{array}
\end{equation}
Differentiating $V(t)$ in time along the trajectories of \eqref{eq_cl_loop}, we have
\begin{equation}
\label{eq_L_F_dif1}
\begin{array}{l}
\dot{V}(t) + 2 \delta V(t)=
2\int_{0}^{l}[z(x,t)\frac{\partial}{\partial x} [a_1(x) z_{x}(x,t)] 
\\
~~~~~~
+ a_2(x)z_x(x,t)z(x,t)+2 \delta \int_{0}^{l} z^2(x,t) dx
\\
~~~~~~
-(K-\phi(z,x,t)) z^2(x,t) +z(x,t) f(x,t)]dx 
\\
~~~~~~
+ 2 K \sum_{j=0}^{N-1} \int_{x_j}^{x_{j+1}} z(x,t) [z(x,t) -F^j(x,t)].
\end{array}
\end{equation}
Taking into account the boundary conditions \eqref{eq1_b1} or \eqref{eq1_b2} and integrating by parts the first term in \eqref{eq_L_F_dif1}, we get
\begin{equation}
\label{eq_po_chast}
\begin{array}{l}
2 \int_{0}^{l}  z(x,t)\frac{\partial}{\partial x} [a_1(x) z_{x}(x,t)] dx
\\
= 2a_1(x) z(x,t) z_x(x,t) \Big|_{0}^{l}-2 \int_{0}^{l} a_1(x) z_{x}^2(x,t) dx
\\
\leq -2 \underline{a}_1 \int_{0}^{l} z_{x}^2(x,t) dx.
\end{array}
\end{equation}

Using Young's inequality for the penultimate term in \eqref{eq_L_F_dif1}, we obtain
\begin{equation}
\label{eq_W_I1}
\begin{array}{l}
2 K \sum_{j=0}^{N-1} \int_{x_j}^{x_{j+1}} z(x,t) [z(x,t) - F^j(x,t)] dx
\\
\leq 
K \bar{R} \int_{0}^{l} z^2(x,t) dx
\\
+ K \bar{R}^{-1} \sum_{j=0}^{N-1} \int_{x_j}^{x_{j+1}} [z(x,t) - F^j(x,t)]^2 dx.
\end{array}
\end{equation}

Applying Lemma \ref{lem_Har_ext} to \eqref{eq_W_I1} and taking into account $z(\bar{x}_j,t)=F^j(\bar{x}_j,t)$, one gets

\begin{equation}
\label{eq_W_I55}
\begin{array}{l}
K \bar{R}^{-1} \sum_{j=0}^{N-1} \int_{x_j}^{x_{j+1}} [z(x,t) - F^j(x,t)]^2 dx \leq
\\
\leq
\frac{4 \Delta^2}{\pi^2} K \bar{R}^{-1} \sum_{j=0}^{N-1} \int_{x_j}^{x_{j+1}}  [z_x^2(x,t) 
\\
- 2 z_x(x,t) F_x^j(x,t) + (F_x^j(x,t))^2] dx.
\end{array}
\end{equation}

Denote by $\eta_j=col\{z(x,t), z_x(x,t),f(x,t), F_x^j(x,t)\}$. Applying \eqref{eq_po_chast}--\eqref{eq_W_I55} to \eqref{eq_L_F_dif1}, we have

\begin{equation}
\label{eq_Lyap1}
\begin{array}{l}
\dot{V}(t)+ 2\delta V(t) -\beta_1 \int_{0}^l f^2(x,t) dx
\\
~~~~~~~~~~~
 - \beta_2  \sum_{j=0}^{N-1} \int_{x_j}^{x_{j+1}} F_x^j(x,t) dx 
\\
~~~~~~~~~~~
 \leq 
 \sum_{j=0}^{N-1} \int_{x_j}^{x_{j+1}} \eta^{\rm T} \Psi \eta dx,
\end{array}
\end{equation}
where $\Psi$ is given by \eqref{eq_LMI}. 
The expression \eqref{eq_LMI} is affine with respect to $a_2$.
According to\cite{Fridman14}, if LMIs \eqref{eq_LMI0} holds in the vertices $a_2=\underline{a}_2$ and $a_2=\overline{a}_2$, then LMI $\Psi \leq 0$ holds for any $a_2 \in [\underline{a}_2,\overline{a}_2]$.
Therefore, the inequality
\begin{equation}
\label{eq_Nerav}
\begin{array}{l}
\dot{V}(t)+ 2\delta V(t) -\beta_1 \int_{0}^l f^2(x,t) dx - 
\\
\beta_2  \sum_{j=0}^{N-1} \int_{x_j}^{x_{j+1}} F_x^j(x,t) dx \leq 0
\end{array}
\end{equation} 
holds. 
Using Lemma \ref{lem_Hal}, the solution of the differential inequality \eqref{eq_Nerav} can be defined as 
\begin{equation}
\label{eq_solution_V}
\begin{array}{l}
V(t) \leq V(0) e^{-2 \delta t} + \frac{\gamma}{2 \delta}.
\end{array}
\end{equation}
Expression \eqref{eq_th2} follows from \eqref{eq_solution_V}.
Theorem \ref{Th1} is proved.

\end{proof}


\section{Main result for a system of hyperbolic type \eqref{eq01}} \label{Sec04}

Substituting \eqref{eq_g_c_l} into \eqref{eq01}, rewrite the closed-loop system in the form
\begin{equation}
\label{eq_cl_loop0}
\begin{array}{l}
z_{tt}(x,t)=\frac{\partial}{\partial x} [a_1(x) z_{x}(x,t)] + a_2(x) z_x(x,t)
\\
~~~~~~~~~~~~~
- b(z,x,t) z_t(x,t) +f(x,t)
\\
~~~~~~~~~~~~~  
- (K-\phi(z,x,t)) z(x,t) 
\\
~~~~~~~~~~~~~
+ K[z(x,t) - F^j(x,t)],
\\
~~~~~~~~~~~~~~ x \in [x_{j}, x_{j+1}), ~~~ j=0,...,N-1.
\end{array}
\end{equation}

\begin{theorem}
\label{Th2}
Consider the closed-loop system \eqref{eq_cl_loop0} under the boundary conditions \eqref{eq1_b1} or \eqref{eq1_b2}.
Given $p \in (-0.5;0.5)$, $\bar{R}>0$, $\Delta>0$, and $\delta>0$ there exists $K>0$ such that
the following LMIs hold
\begin{equation}
\label{eq_LMI00}
\begin{array}{l}
\bar{\Psi} \left( \phi=\{\underline{\phi},\overline{\phi}\}, a_2=\{\underline{a}_2,\overline{a}_2\}, b=\{\underline{b},\overline{b}\} \right) \leq 0,
\end{array}
\end{equation}
where
\small
\begin{equation}
\label{eq_LMI000}
\begin{array}{l}
\bar{\Psi}=
\begin{bmatrix}
\bar{\Psi}_{11} & 0.5p a_2 & \bar{\Psi}_{13} & 0.5 p & 0\\
* & \bar{\Psi}_{22} & a_2 & 0 & -\frac{4 \Delta^2}{\pi^2} K \bar{R}^{-1}\\
* & * & \bar{\Psi}_{33}  & 1 & 0 \\
* & * & * & -\beta_1 & 0 \\
* & * & * & * & -\beta_2 + \frac{4 \Delta^2}{\pi^2} K \bar{R}^{-1}
\end{bmatrix},
\\
\\
\bar{\Psi}_{11}=-0.5 p (K- \overline{\phi}) + 0.25 K \bar{R} p^2 + 2\delta,
\\
\bar{\Psi}_{13} = 1-0.5 p b -K+ \phi + K \bar{R}p,
\\
\bar{\Psi}_{22} = - p \underline{a}_1 + \frac{4 \Delta^2}{\pi^2} K \bar{R}^{-1},
\\
\bar{\Psi}_{33} = 0.5 p - 2 \underline{b} + 0.5 K \bar{R}.
\end{array}
\end{equation}
\normalsize
Then inequality \eqref{eq_th2} holds by taking into account parameters from \eqref{eq_LMI00}.
\end{theorem}


\begin{proof} 
Introduce the following Lyapunov functional
\begin{equation}
\label{eq_L_F0}
\begin{array}{l}
V(t)=\int_{0}^{l} [\underline{a} z^2_{x}(x,t) + z^2(x,t)
\\
~~~~~~~~~
 + p z(x,t)z_t(x,t) + z_t^2(x,t)] dx.
\end{array}
\end{equation}
The inequality $z^2 + p zz_t + z_t^2 \geq 0$ holds for $p \in (-0.5;0.5)$. Therefore, $V(t) \geq 0$. 
Differentiating $V(t)$ in time along the trajectories of \eqref{eq_cl_loop0}, we have
\begin{equation}
\label{eq_L_F_dif10}
\begin{array}{l}
\dot{V}(t) + 2 \delta V(t)=
2\int_{0}^{l}[a_1(x) z_x(x,t) z_{xt}(x,t) 
\\~~~~~~
+ z(x,t) z_t(x,t) + 0.5 p z_t^2(x,t) +\\~~~~~~

[0.5 p z(x,t) + z_t(x,t)]
[\frac{\partial}{\partial x} [a_1(x) z_{x}(x,t)] + 
\\~~~~~~
a_2(x) z_x(z,x,t)
-b(z,x,t) z_t(x,t)
\\~~~~~~
+f(x,t)-(K-\phi) z(x,t)] dx
\\~~~~~~
+ 2 K \sum_{j=0}^{N-1} \int_{x_j}^{x_{j+1}} [0.5 p z(x,t) + z_t(x,t)]
\\~~~~~~
\times [z(x,t) -F^j(x,t)]
+2 \delta \int_{0}^{l} z^2(x,t) dx.
\end{array}
\end{equation}

Taking into account the boundary conditions \eqref{eq1_b1} or \eqref{eq1_b2}, consider the integration by parts
\begin{equation}
\label{eq_po_chast01}
\begin{array}{l}
2 \int_{0}^{l}  z_t(x,t) \frac{\partial}{\partial x} [a_1(x) z_{x}(x,t)] dx
\\
=2 a_1(x) z_x(x,t) z_t(x,t) \Big|_{0}^{l}
\\
-2 \int_{0}^{l} a_1(x) z_{x}(x,t) z_{xt}(x,t) dx 
\\
=-2 \underline{a}_1 \int_{0}^{l} z_{x}(x,t) z_{xt}(x,t) dx.
\end{array}
\end{equation}

Using Young's inequality for the penultimate term in \eqref{eq_L_F_dif10}, we obtain
\begin{equation}
\label{eq_W_I10}
\begin{array}{l}
2 K \sum_{j=0}^{N-1} \int_{x_j}^{x_{j+1}} [0.5 p z(x,t) + z_t(x,t)]
\\~~~~
\times [z(x,t) - F^j(x,t)] dx \leq
\\~~~~
\leq 
K \bar{R} \int_{0}^{l} [0.5 p z(x,t) + z_t(x,t)]^2 dx
\\~~~~
+ K \bar{R}^{-1} \sum_{j=0}^{N-1} \int_{x_j}^{x_{j+1}} [z(x,t) - F^j(x,t)]^2 dx.
\end{array}
\end{equation}

Considering $z(\bar{x}_j,t)=F^j(\bar{x}_j,t)$ and applying Lemma \ref{lem_Har_ext} to \eqref{eq_W_I10}, we have \eqref{eq_W_I55}.
Denote by $\bar{\eta}_j=col\{z(x,t), z_x(x,t), z_t(x,t), f(x,t), F_x^j(x,t)\}$. Applying \eqref{eq_po_chast}, \eqref{eq_po_chast01}, \eqref{eq_W_I10} and \eqref{eq_W_I55} to \eqref{eq_L_F_dif10}, rewrite result as follows
\begin{equation}
\label{eq_Lyap1}
\begin{array}{l}
\dot{V}(t)+ 2\delta V(t) -\beta_1 \int_{0}^l f^2(x,t) dx 
\\~~~~~~
- \beta_2  \sum_{j=0}^{N-1} \int_{x_j}^{x_{j+1}} F_x^j(x,t) dx
\\~~~~~~
 \leq 
 \sum_{j=0}^{N-1} \int_{x_j}^{x_{j+1}} \bar{\eta}^{\rm T} \bar{\Psi} \bar{\eta} dx,
\end{array}
\end{equation}
where $\bar{\Psi}$ is given by \eqref{eq_LMI000}.
The expression \eqref{eq_LMI000} is affine with respect to the parameters $\phi$, $a_2$ and $b$. 
According with\cite{Fridman14}, if LMIs \eqref{eq_LMI0} hold in the vertices $\phi=\{\underline{\phi},\overline{\phi}\}$, $a_2=\{\underline{a}_2,\overline{a}_2\}$, $b=\{\underline{b},\overline{b}\}$, then LMI $\bar{\Psi} \leq 0$ holds for any $\phi=[\underline{\phi},\overline{\phi}]$, $a_2=[\underline{a}_2,\overline{a}_2]$, $b=[\underline{b},\overline{b}]$.
Therefore, inequality \eqref{eq_Nerav} holds.
Considering Lemma \ref{lem_Hal}, the solution of differential inequality \eqref{eq_Lyap1} can be defined as
\eqref{eq_solution_V}.
Then, inequality \eqref{eq_th2} follows from \eqref{eq_solution_V} but taking into account parameters from \eqref{eq_LMI00}.
Theorem \ref{Th2} is proved.
\end{proof}


\section{Well-posedness of the closed-loop system} \label{Sec5}

In this section, we show that there exist solutions of the closed-loop systems \eqref{eq_cl_loop} and \eqref{eq_cl_loop0} satisfying Dirichlet boundary conditions \eqref{eq1_b1}. The well-posedness under the mixed conditions \eqref{eq1_b2} can be proved similarly.

\subsection{The closed-loop system \eqref{eq_cl_loop}}
\label{KZ}

Consider the closed-loop system \eqref{eq_cl_loop} with the boundary conditions \eqref{eq1_b1}. The boundary-value problem \eqref{eq_cl_loop}, \eqref{eq1_b1} can be represented as an abstract inhomogeneous Cauchy problem in the Hilbert space $H=L_2(0,l)$ as follows
\begin{equation}
\label{eq_new1}
\begin{array}{l}
\dot{z}(t) = \mathcal{A} z(t) +F(t,z(t)), ~~ z_0=z(0) \in \mathcal{D(A)}.
\end{array}
\end{equation}
Here the operator
$\mathcal{A} = \frac{\partial}{\partial x}[a_1(x)\frac{\partial}{\partial x}] + a_2(x)\frac{\partial}{\partial x}$
has the dense domain 
$\mathcal{D(A)} = \{z \in H_2(0,l): z(0)=z(l)=0\}$,
$F(t,z(t))=f(t)+u(t)+\phi(t,z(t))z(t)$, the function $u(t)$ is given by  \eqref{eq_g_c_l}.
According to Theorem \ref{Th1}, \cite{Henry93} and \cite{Curtain95}, infinitesimal operator $\mathcal{A}$ generates a strictly continuous exponentially stable semigroup ($C_0$-semigroup) $T(t)$. 
Then, the boundary-value problem \eqref{eq_new1} can be represented as a boundary-value problem on the semi-infinite interval $[0,\infty)$ and its solutions can be found as solutions of the following integral equation
\begin{equation}
\label{eq_new2}
\begin{array}{l}
z(t) =T(t) z(0) + \int_{0}^{t} T(t-s )F(s,z(s))ds.
\end{array}
\end{equation}
Since the function $F(t,z(t))$ is continuously differentiable w.r.t. $t$, then, according to Theorem 3.1.3 \cite{Curtain95}, there is a unique solution of
\eqref{eq_new1} and this solution satisfies the integral equation \eqref{eq_new2}.


\subsection{The closed-loop system \eqref{eq_cl_loop0}}

Now we consider the closed-loop system \eqref{eq_cl_loop0} with the boundary conditions \eqref{eq1_b1}.
Rewrite the boundary-value problem \eqref{eq_cl_loop0}, \eqref{eq1_b1} as an abstract inhomogeneous Cauchy problem in the Hilbert space $H=L_2(0,l)$ in the form
\begin{equation}
\label{eq_new10}
\begin{array}{l}
\dot{\xi}(t) = \bar{\mathcal{A}} \xi(t) +\bar{F}(t,\xi(t)), ~~ \xi_0=\xi(0) \in \mathcal{\bar{\mathcal{A}}},
\end{array}
\end{equation}
where $\xi=col\{z,z_t\}$, the operator
$\bar{\mathcal{A}} = 
\begin{bmatrix}
0 & 1 \\
\frac{\partial}{\partial x}[a_1(x)\frac{\partial}{\partial x}]+a_2(x)\frac{\partial}{\partial x} & 0
\end{bmatrix}$
has the dense domain
$\mathcal{D(\bar{\mathcal{A}})} = \{\xi \in H_2(0,l): \xi(0)=\xi(l)=0\}$,
$\bar{F}(t,\xi(t))=
[0~1]^{\rm T}
\left[ f(t)+u(t) + \phi(t,\xi_1(t))\xi_1(t) - b(t,\xi_1(t))\xi_2(t)\right]$, the function $u(t)$ is given by \eqref{eq_g_c_l}. 
According to Theorem \ref{Th1}, \cite{Henry93} and \cite{Curtain95}, the infinitesimal operator $\bar{\mathcal{A}}$ generates a strictly continuous exponentially stable semigroup ($C_0$-semigroup) $T(t)$. Therefore, further considerations for \eqref{eq_new10} are similar to those for \eqref{eq_new1} in  Subsection \ref{KZ}.


\section{Examples}
\label{Examples}


\subsection{The simulations of systems \eqref{eq1} and \eqref{eq01}}

Let $l=1$. 
To simulate systems \eqref{eq1} and \eqref{eq01} we divide the segment $[0,1]$ into $160$ sub-intervals of the same length with a sampling step in the spatial variable $D=1/160$.
The first and the second derivatives in the spatial variable of $z(x,t)$ are calculated in the points $0=x_0<x_1<...x_l<...<x_{160}$ according with the following expressions $z_{x}(x_k,t) = \frac{z(x_{k+1},t)-z(x_k,t)}{D}$ and $z_{xx}(x_k,t) = \frac{z(x_{k+1},t)-2z(x_k,t)+z(x_{k-1},t)}{D^2}$.

For design the control law \eqref{eq_g_c_l} we divide the segment $[0,1]$ into $N=2$ and $N=10$ equal sub-intervals (see \eqref{eq1_sub_int}).

Consider systems \eqref{eq1} and \eqref{eq01} under the Dirichlet boundary conditions \eqref{eq1_b1} and $a_1(x) \in [0.5,2]$, $a_2(x) \in [-5,5]$, $\phi(z,x,t)  \in [-5,5]$, $b(z,x,t)=[-5,-1]$, $|f(x,t)| \leq 20$ for any $x$ and $t$. The initial conditions are given as follows $z(x,0)=\sin(\pi x)$ è $z_x(x,0)=0$.

The matrix inequalities \eqref{eq_LMI0} and \eqref{eq_LMI00} are feasiable for $K \geq 100$.


\subsection{The results of simulations of system \eqref{eq1}}

In \eqref{eq1} choose $a_1(x) =1+\sin(x)$, $a_2(x) =-2-\sin(1.3x)$, $\phi(z,x,t)  = 5+\cos(3z)$, and $f(x,t)=0.2[\sin(30t)+\sin(2t)]$. 
Consider two partitions for $N=2$ and $N=10$ (see \eqref{eq1_sub_int}). 
Let $\bar{x}_j=0.5(x_{j} + x_{j+1})$. 
Choose $F^j(x,t)$ in the control law \eqref{eq_g_c_l} from example 3, where $\alpha=10^3$, as well as $\beta_j=1/8$ for $N=2$ and $\beta_j=1/80$ for $N=10$. 
In Figs.~\ref{Fig_Fridman}--\ref{Fig_My_Big} the solutions of \eqref{eq1} and the spatio-temporal graphs of $u(x,t)$ are illustrated for:
\begin{enumerate}
\item[1)] the control law $u(x,t)=-K z(\bar{x}_j,t)$ from \cite{Fridman12,Fridman14,Fridman20} for $K=100$;
\item[2)] the proposed control law \eqref{eq_g_c_l} for $K=-100$;
\item[3)] the proposed control law \eqref{eq_g_c_l} for $K=-500$.
\end{enumerate}
Figs.~\ref{Fig_Fridman}, \ref{Fig_My} show, that the solutions of $z(x,t)$ almost the same for the proposed control law and the one from \cite{Fridman12}.
However, the proposed control law provides exponential stability under distributed disturbances.
If the coefficient $K$ is increased by $5$ times, then the magnitude of the control signal is also increased approximately $5$ times.
In this case the rate of exponential convergence and quality of disturbance rejection in the steady state are higher than the ones from \cite{Fridman12} at $K=100$.
 
Now we analyze the control costs. Figs.~\ref{Fig_Energy1},~\ref{Fig_Energy2} illustrate the integral difference in the form $I=\sum_{j=0}^{N-1} \int_{0}^{t} \left(|u_{F\&B}(\bar{x}_j,s)|-|u_{proposed}(\bar{x}_j,s)| \right)ds$, where $u_{F\&B}(\bar{x}_j,t)$ is the control law from \cite{Fridman12}, $u_{proposed}(\bar{x}_j,t)$ is the proposed one. 
The advantages of the proposed algorithm is clearly seen, i.e. the control costs of the proposed control law are less than ones from \cite{Fridman12}. 
Moreover, the proposed control law can be approximated by finite actions, while the control law from \cite{Fridman12,Fridman14,Fridman20} requires implementation throughout the whole spatial variable.
The simulations for the function $F^j(x,t)$ from example 2 with $\alpha=100$ are comparable with the results obtained for the function $F^j(x,t)$ from example 3.

\begin{figure}[h]
\center{\includegraphics[width=1\linewidth]{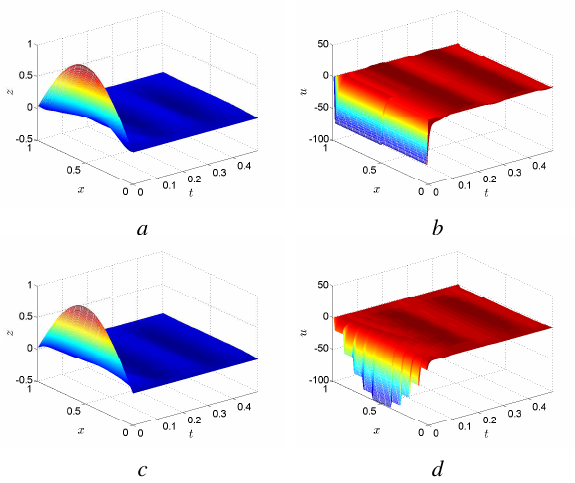}}
\caption{The spatio-temporal graphs of $z(x,t)$ and $u(x,t)$ from \cite{Fridman12,Fridman14,Fridman20} for $N=2$ (\textit{a,b}) and $N=10$ (\textit{c,d}) for $K=100$.}
\label{Fig_Fridman}
\end{figure}

\begin{figure}[h]
\center{\includegraphics[width=1\linewidth]{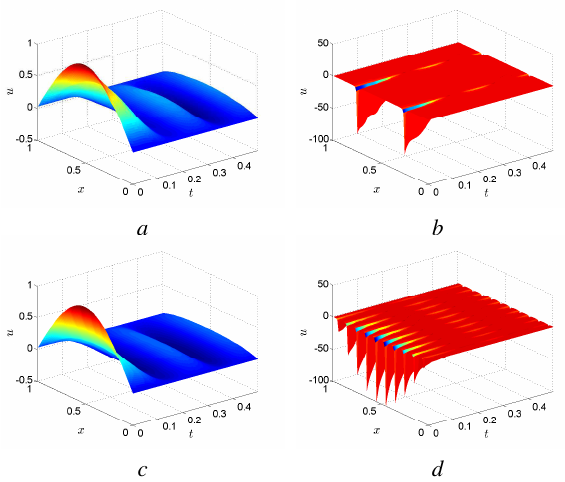}}
\caption{The spatio-temporal graphs of $z(x,t)$ and $u(x,t)$ for the proposed control law for $N=2$ (\textit{a,b}), $N=10$ (\textit{c,d}) and $K=100$.}
\label{Fig_My}
\end{figure}

\begin{figure}[h]
\center{\includegraphics[width=1\linewidth]{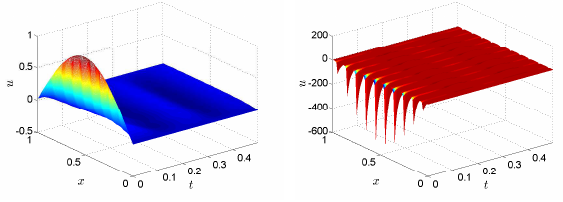}}
\caption{The spatio-temporal graphs of $z(x,t)$ è $u(x,t)$ for the proposed control law for $N=10$ and $K=500$.}
\label{Fig_My_Big}
\end{figure}

\begin{figure}[h]
\center{\includegraphics[width=0.9\linewidth]{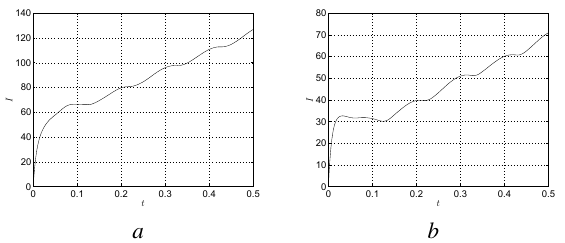}}
\caption{The plots of integral difference for $N=2$  (\textit{a}) and $N=10$ (\textit{b}) between the control law \cite{Fridman12,Fridman20} and the proposed control law for $K=100$.}
\label{Fig_Energy1}
\end{figure}

\begin{figure}[h]
\center{\includegraphics[width=0.4\linewidth]{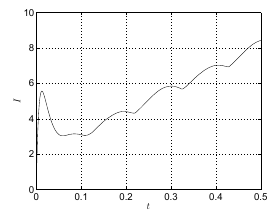}}
\caption{The plots of integral difference for $N=10$ between the control law \cite{Fridman12,Fridman14,Fridman20} and the proposed control law for $K=100$.}
\label{Fig_Energy2}
\end{figure}


\subsection{The results of simulations of system \eqref{eq01}}

In \eqref{eq01} choose $a_1(x) =1+\sin(x)$, $a_2(x) =-2-\sin(1.3x)$, $b(z,x,t)=-2-\sin(1.1z)$, $\phi(z,x,t)  = 5+\cos(3z)$ and $f(x,t)=0.2[\sin(30t)+\sin(2t)]$.
In Figs.~\ref{Fig_Wave1}, \ref{Fig_Wave2} the solutions of \eqref{eq01} and spatio-temporal graphs of $u(x,t)$ are given for the proposed control law \eqref{eq_g_c_l} with the function $F^j(x,t)$ from example 3 for $\alpha=10^3$, as well as $\beta_j=1/8$ for $N=2$ and $\beta_j=1/80$ for $N=10$ and $K=-100$, $K=-500$.
 The proposed control law provides exponential stability under disturbances. 
 The simulations for the function $F^j(x,t)$ from example 2 with $\alpha=100$ are comparable with the results obtained for the function $F^j(x,t)$ from example 3.

\begin{figure}[h]
\center{\includegraphics[width=1\linewidth]{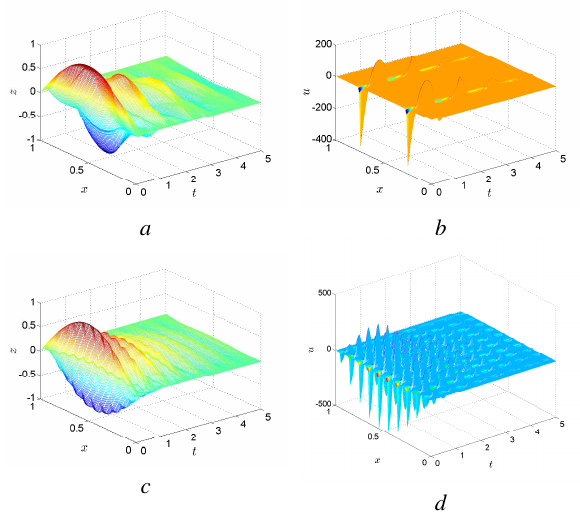}}
\caption{The spatio-temporal graphs of $z(x,t)$ and $u(x,t)$ for the proposed control law for $N=2$ (\textit{a, b}), $N=10$ (\textit{c, d}) and $K=100$.}
\label{Fig_Wave1}
\end{figure}

\begin{figure}[h]
\center{\includegraphics[width=1\linewidth]{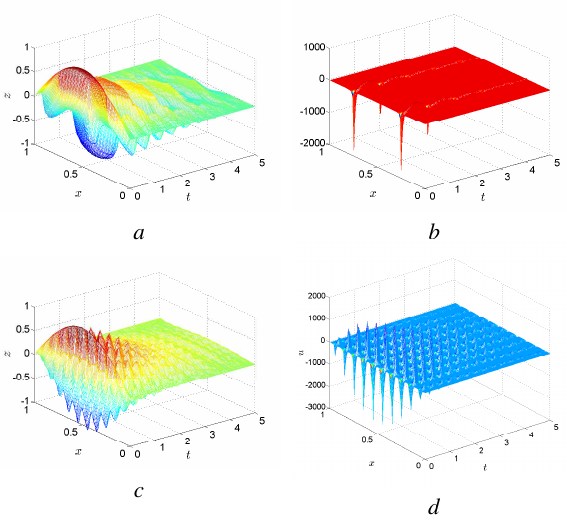}}
\caption{The spatio-temporal graphs of $z(x,t)$ and $u(x,t)$ for the proposed control law for $N=2$ (\textit{a, b}), $N=10$ (\textit{c, d}) and $K=500$.}
\label{Fig_Wave2}
\end{figure}


\section{Conclusions} \label{Sec7}

A sampled-data in space control law is proposed for scalar semilinear differential equations of parabolic and hyperbolic types with interval-indefinite parameters and external bounded disturbances.
The control law is only used a finite set of measurements of the output signal.
Also the control law depends on a function depending on the spatial coordinate and the current measurement.
This function allows one to achieve different properties, for example, to provide reduced control costs.
The exponential stability of the closed-loop systems and robustness with respect to unknown parameters and external disturbances are considered.
The simulations have confirmed the theoretical results and have showed the efficiency of the proposed algorithm compared with ones from \cite{Fridman12,Fridman14,Fridman20}.


%

\section*{Acknowledgment}

The results were proposed with the support of a grant from the Russian Science Foundation (No. 18-79-10104) in IPME RAS.

$ $

$ $

\end{document}